\documentclass[a4paper,10pt]{article}
\usepackage[a4paper, left=25mm, right=25mm, top=25mm, bottom=25mm]{geometry}
\usepackage[T1]{fontenc}
\usepackage{lmodern}
\usepackage[utf8]{inputenc}
\usepackage{amsmath, amssymb, amsthm, bbm, mathtools, nccmath}
\usepackage{xparse}
\usepackage[table]{xcolor}
\usepackage{url, hyperref}
\usepackage{subcaption, caption}
\usepackage{tikz}
\theoremstyle{plain} 
    \newtheorem{theorem}{Theorem}
    \newtheorem{lemma}[theorem]{Lemma}
    \newtheorem*{lemma*}{Lemma}
    \newtheorem{corollary}[theorem]{Corollary}
    \newtheorem{proposition}[theorem]{Proposition}

\theoremstyle{definition}

\theoremstyle{remark}
    \newtheorem{remark}[theorem]{Remark}
\renewcommand{\mathbf}{\boldsymbol}

\renewcommand{\ge}{\geqslant}
\renewcommand{\geq}{\geqslant}
\renewcommand{\le}{\leqslant}
\renewcommand{\leq}{\leqslant}
\renewcommand{\tilde}{\widetilde}
\newcommand*{\set}[1]{\mathcal{#1}} 
\newcommand*{\fnc}[1]{\mathrm{#1}} 
\newcommand*{\rv}[1]{\mathsf{#1}} 
\newcommand*{\rvs}[1]{\mathbf{\mathsf{#1}}} 
\newcommand*{\sys}[1]{\mathsf{#1}} 
\newcommand*{\syss}[1]{\mathbf{\mathsf{#1}}} 
\newcommand\diff{\mathop{}\!\mathrm{d}} 

\newcommand*{\tensor}{\otimes}

\DeclareMathOperator{\tr}{tr}

\DeclareMathOperator{\Cov}{Cov}

\newcommand{\hilbert}{\mathcal{H}}

\newcommand{\id}{\fnc{id}}

\newcommand{\integers}{\mathbb{Z}}
\newcommand{\densop}{\set{L}_1^+}
\newcommand{\posop}{\set{L}^+}
\newcommand{\linop}{\set{L}}
\newcommand{\herm}{\dagger}
\newcommand{\trans}{\mathsf{T}}
\newcommand{\defeq}{\coloneqq}
\newcommand{\defas}{\eqqcolon}

\DeclarePairedDelimiterX{\ket}[1]{\lvert}{\rangle}{#1}
\DeclarePairedDelimiterX{\bra}[1]{\langle}{\rvert}{#1}
\DeclarePairedDelimiterX{\spr}[2]{\langle}{\rangle}{#1\delimsize\vert#2}
\DeclarePairedDelimiterX{\proj}[1]{\lvert}{\rvert}{#1\delimsize\rangle\!\delimsize\langle#1}
\DeclarePairedDelimiterX{\floor}[1]{\lfloor}{\rfloor}{#1}
\DeclarePairedDelimiterX{\ceil}[1]{\lceil}{\rceil}{#1}
\DeclarePairedDelimiterX{\abs}[1]{\lvert}{\rvert}{#1}
\DeclarePairedDelimiterX{\norm}[1]{\lVert}{\rVert}{#1}
\DeclarePairedDelimiterX{\size}[1]{\lvert}{\rvert}{#1}
\DeclarePairedDelimiterX{\infdiv}[2]{(}{)}{#1\delimsize\Vert#2}
\DeclarePairedDelimiterX{\inner}[2]{\langle}{\rangle}{#1,#2}
\ExplSyntaxOn
\NewDocumentCommand{\multiadjustlimits}{m}
 {
  \group_begin:
  \multiadjustlimits_measure:n { #1 }
  \multiadjustlimits_print:n { #1 }
  \group_end:
 }

\tl_new:N  \l__multiadjustlimits_operator_tl
\tl_new:N  \l__multiadjustlimits_limit_tl

\cs_new_protected:Nn \multiadjustlimits_measure:n
 {
  \clist_map_function:nN { #1 } \__multiadjustlimits_measure:n
 }
\cs_new_protected:Nn \__multiadjustlimits_measure:n
 {
  \__multiadjustlimits_measure:NNn #1
 }
\cs_new_protected:Nn \__multiadjustlimits_measure:NNn
 {
  \tl_put_right:Nn \l__multiadjustlimits_operator_tl { #1 }
  \tl_put_right:Nn \l__multiadjustlimits_limit_tl { #3 }
 }

\cs_new_protected:Nn \multiadjustlimits_print:n
 {
  \clist_map_function:nN { #1 } \__multiadjustlimits_print:n
 }
\cs_new_protected:Nn \__multiadjustlimits_print:n
 {
  \__multiadjustlimits_print:NNn #1
 }
\cs_new_protected:Nn \__multiadjustlimits_print:NNn
 {
  \mathop { \vphantom{\l__multiadjustlimits_operator_tl} \mathopen{} #1 }
  \limits
  \sb{ \vphantom{\cramped{\l__multiadjustlimits_limit_tl}} #3 }
 }

\ExplSyntaxOff

\newcommand\ie{\textit{i.e.}}

\newcommand\cf{\textit{cf.}}
\newcommand\wrt{w.r.t.~}

\begin{document}
\title{\Large \bfseries Channel coding against quantum jammers via minimax}
\author{
    Michael X. Cao, Yongsheng Yao, and Mario Berta\\
    {\small Institute for Quantum Information, RWTH Aachen University, Germany}
}
\date{\vspace{-1em}}
\maketitle
\begin{abstract}
We study classical communication over fully quantum arbitrarily varying channels, in which a jammer can prepare an arbitrary state entangled across channel uses.
We formulate assisted codes as joint input-output maps and apply Sion's minimax theorem to exchange optimization over codes and jamming states at fixed message size.
Combining this reduction with one-shot coding bounds and a generalized quantum asymptotic equipartition property yields a single-letter formula for the entanglement-assisted classical capacity, a single-letter formula for the shared-randomness-assisted capacity of classical-quantum channels, and a regularized formula for general quantum input channels.
For entanglement assistance and for shared-randomness-assisted classical-quantum channels, we extend the capacity formulas to arbitrary jammer Hilbert spaces while keeping the sender and receiver systems finite-dimensional.
The extension uses finite-dimensional approximations at each blocklength, a uniform bound on the change in coding error, and compact closures of finite-dimensional output-state sets.
These arguments avoid de Finetti reductions and establish equality with the corresponding compound-channel capacities in these two settings.
\end{abstract}
\section{Introduction}
Communication in the presence of an adversary, or \emph{jammer}, has been studied since the late 1950s~\cite{blackwell1959capacity, blackwell1960capacities}.
A fixed but unknown channel is modeled as a compound channel~\cite{wolfowitz1961coding} whereas variation between channel uses leads to arbitrarily varying channels (AVCs)~\cite{ahlswede1970note} (see~\cite{csiszar2011information} for a textbook).
With free shared randomness between the sender and receiver, the two models have the same channel-coding capacity (assuming they arise from the same \emph{convex} set of channels)~\cite{blackwell1960capacities}.
Furthermore, for compound channels, shared randomness does not increase asymptotic capacity~\cite{ahlswede1978eliminierung}, \ie,
\begin{equation}\label{eq:avc=comp:SR}
C_\fnc{avc}^{\fnc{SR}}(W_{\rv{B}|\rv{AE}}) \overset{\text{convexity assumption}}{=} C_\fnc{comp}^{\fnc{SR}}(W_{\rv{B}|\rv{AE}}) = C_\fnc{comp}(W_{\rv{B}|\rv{AE}}).
\end{equation}
Shared randomness can increase the capacity of AVCs. In particular, there exist channels for which $C_\fnc{avc}^\fnc{SR}>0$ but $C_\fnc{avc}=0$~\cite{blackwell1960capacities} (and \cite[Example~2]{csiszar1988capacity}). Combined with~\eqref{eq:avc=comp:SR}, this shows that channel variation can reduce the capacity without shared randomness.
Note that throughout this paper, we at least allow the adversary (Eve) to generate the jamming signal stochastically at \emph{each} channel use. This assumption ensures that the set of channels considered under the compound channel or arbitrarily varying channel models remains convex.

Adversarial communication models have also been studied in quantum information theory.
The classical capacity of arbitrarily varying classical-quantum (CQ) channels was investigated in~\cite{ahlswede2007classical}, while the quantum capacity of arbitrarily varying quantum channels was addressed in~\cite{ahlswede2013quantum}. 
Furthermore, the entanglement-assisted classical capacities of compound and arbitrarily varying quantum channels were studied in~\cite{berta2017entanglement, boche2017entanglement}. 
In particular, it is known that~\cite{boche2017entanglement},
\begin{equation}
C_\fnc{avc}^\fnc{EA}(\mathcal{N}_{\sys{AE}\to\sys{B}}) = \underbrace{\adjustlimits\sup_{\rho_{\sys{A}'}} \inf_{\sigma_\sys{E}} I(\sys{A}':\sys{B})_{(\id_{\sys{A}'}\tensor \mathcal{N}_{\sys{AE}\to\sys{B}})(\proj{\rho}_{\sys{A}'\sys{A}}\tensor \sigma_{\sys{E}})}}_{\defas I_{\fnc{adv.}\sys{E}}(\mathcal{N}_{\sys{AE}\to\sys{B}})}
\overset{\text{convexity assumption}}{=}  C_\fnc{comp}^\fnc{EA}(\mathcal{N}_{\sys{AE}\to\sys{B}})
\end{equation}
where $I$ denotes the quantum mutual information, and $\ket{\rho}_{\sys{A}'\sys{A}}$ the canonical purification of $\rho_{\sys{A}'}$, \ie, $\ket{\rho}_{\sys{A}'\sys{A}} = (\sqrt{\rho_{\sys{A}'}}\tensor I_{\sys{A}})  \ket{\gamma}_{\sys{A}'\sys{A}}$ with $\ket{\gamma}_{\sys{A}'\sys{A}}$ being the unnormalized maximally entangled vector $\sum_i\ket{i}\tensor\ket{i}$ on systems $\sys{A}'\sys{A}$ with $\sys{A}'\equiv\sys{A}$. However, in all of the above setups, the channel variability is driven by a classical adversary -- that is, the jammer is at most equipped with classical memory.\footnote{Even when the jamming interface is quantum, if the jammer inputs are restricted to separable states, the resulting model is effectively equivalent to a (classical) AVC, and there is no need to treat the jammer inputs as genuinely quantum.}

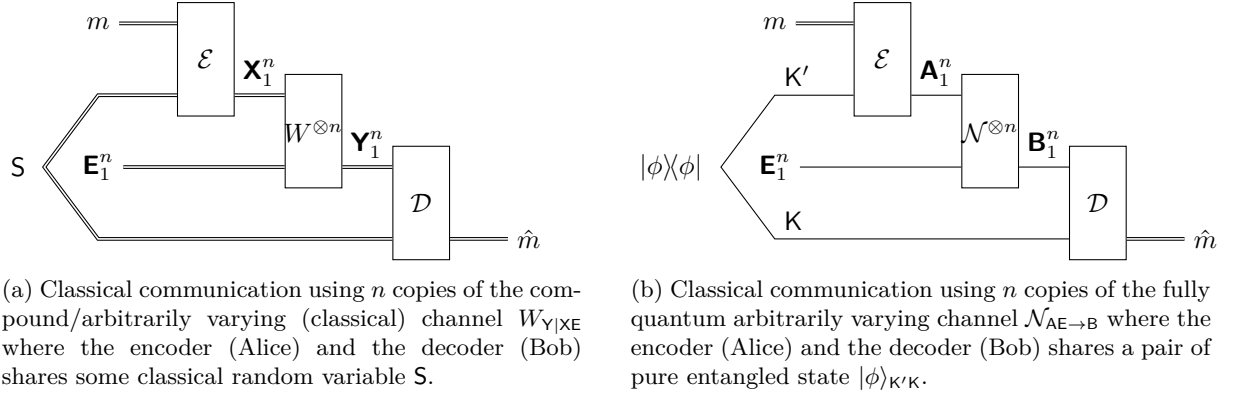
\begin{figure}[t]
\centering
\begin{subfigure}[t]{.48\textwidth}
\begin{tikzpicture}[x = 1.5cm, y = -1cm, scale=.91]
    \node (M) at (0,0) {$m$};
    \node (K') at (0,1) {};
    \node (E) at (0,2) {$\syss{E}_1^n$};
    \node (K) at (0,3) {};
    \node (KK') at (-0.5,2) {};
    \node (enc) [rectangle, draw, minimum height=1.5cm, minimum width = 0.75cm] at (1,0.5) {$\mathcal{E}$};
    \node (channels) [rectangle, draw, minimum height=1.5cm, minimum width = 0.75cm] at (2,1.5) {}; \node at (channels) {$W^{\tensor n}$};
    \node (dec) [rectangle, draw, minimum height=1.5cm, minimum width = 0.75cm] at (3,2.5) {$\mathcal{D}$};
    \node (hatM) at (4,3) {$\hat{m}$};
    
    \draw[double] (enc.west|-K'.center) -- (K'.center) -- (KK'.center) -- (K.center) -- (K.center-|dec.west);
    \node [anchor=east] at (KK'.west) {$\rv{S}$};

    \draw[double] (M) -- (M-|enc.west);
    \path (enc.east|-K') edge[draw, double] node[above] {$\rvs{X}_1^n$} (channels.west|-K');
    \draw[double] (E) -- (E-|channels.west);
    \path (channels.east|-E) edge[draw, double] node[above] {$\rvs{Y}_1^n$} (dec.west|-E);
    \draw[double] (dec.east|-hatM) -- (hatM);
\end{tikzpicture}
\caption{Classical communication using $n$ copies of the compound/arbitrarily varying (classical) channel $W_{\rv{Y}|\rv{XE}}$ where the encoder (Alice) and the decoder (Bob) share some classical random variable $\rv{S}$.}
\label{fig:AVC:n}
\end{subfigure}
\hfill
\begin{subfigure}[t]{.48\textwidth}
\begin{tikzpicture}[x = 1.5cm, y = -1cm, scale=.91]
    \node (M) at (0,0) {$m$};
    \node (K') at (0,1) {};
    \node (E) at (0,2) {$\syss{E}_1^n$};
    \node (K) at (0,3) {};
    \node (KK') at (-0.5,2) {};
    \node (enc) [rectangle, draw, minimum height=1.5cm, minimum width = 0.75cm] at (1,0.5) {$\mathcal{E}$};
    \node (channels) [rectangle, draw, minimum height=1.5cm, minimum width = 0.75cm] at (2,1.5) {}; \node at (channels) {$\mathcal{N}^{\tensor n}$};
    \node (dec) [rectangle, draw, minimum height=1.5cm, minimum width = 0.75cm] at (3,2.5) {$\mathcal{D}$};
    \node (hatM) at (4,3) {$\hat{m}$};
    
    \draw (enc.west|-K'.center) -- (K'.center) -- (KK'.center) -- (K.center) -- (K.center-|dec.west);
    \node [anchor=east] at (KK'.west) {$\proj{\phi}$};
    \node [anchor=south west] at (K'.center) {$\sys{K}'$};
    \node [anchor=south west] at (K.center) {$\sys{K}$};

    \draw[double] (M) -- (M-|enc.west);
    \path (enc.east|-K') edge[draw] node[above] {$\syss{A}_1^n$} (channels.west|-K');
    \draw (E) -- (E-|channels.west);
    \path (channels.east|-E) edge[draw] node[above] {$\syss{B}_1^n$} (dec.west|-E);
    \draw[double] (dec.east|-hatM) -- (hatM);
\end{tikzpicture}
\caption{Classical communication using $n$ copies of the fully quantum arbitrarily varying channel $\mathcal{N}_{\sys{AE}\to\sys{B}}$ where the encoder (Alice) and the decoder (Bob) share a pure entangled state $\ket{\phi}_{\sys{K}'\sys{K}}$.}
\label{fig:FQAVC:n}
\end{subfigure}
\caption{Point-to-point communication with an interfering adversary,
where in both figures $\mathcal{E}$ is the encoder, $\mathcal{D}$ is the decoder, the adversary (Eve) controls the systems $\syss{E}_1^n$. The goal for Alice and Bob is to minimize the probability that $\hat{m}\neq m$ without knowing the jamming state.
}
\label{fig:jamable-channels}
\end{figure}

Recently, the adversarial communication model has been extended to its fully quantum form, giving rise to the concept of \emph{fully quantum arbitrarily varying channels} (FQAVCs)~\cite{boche2018fully}, where the jammer inputs are allowed to be an arbitrary quantum state (see Figure~\ref{fig:FQAVC:n}). 
Interestingly, a similar setup had also been considered in~\cite{karumanchi2016quantum, karumanchi2016classical}, though with a benevolent interference party instead.  
It was shown in~\cite{boche2018fully}, and later in~\cite{belzig2024fully} for entanglement-assisted cases that, given a \emph{finite-dimensional} jammer system, the FQAVC has the same shared-randomness-assisted/entanglement-assisted capacity as the corresponding compound channel.
These results give the capacity equalities
\begin{align}
\label{eq:fqavc=comp:finite:sr}
C_\fnc{fqavc}^\fnc{SR}(\mathcal{N}_{\sys{AE}\to\sys{B}})&\overset{\text{convexity assumption}}{=} C_\fnc{comp}^\fnc{SR}(\mathcal{N}_{\sys{AE}\to\sys{B}}),\\
\label{eq:fqavc=comp:finite:ea}
C_\fnc{fqavc}^\fnc{EA}(\mathcal{N}_{\sys{AE}\to\sys{B}})&\overset{\text{convexity assumption}}{=} C_\fnc{comp}^\fnc{EA}(\mathcal{N}_{\sys{AE}\to\sys{B}}).
\end{align}
Both proofs utilize a de Finetti-type reduction to construct a sequence of codes with diminishing error for the FQAVC $\mathcal{N}_{\sys{AE}\to\sys{B}}$ from codes for the compound channel with the same rate, and thus require the jammer system to be finite dimensional for the error to vanish asymptotically. We summarize the known results in Table~\ref{tab:summary:finite:eve}.
\begin{table}
\centering
\begin{tikzpicture}[x = 1cm, y = -1cm]
    \node[draw=none, inner sep = 0, outer sep = 0] (table) {
    \begin{tabular}{|c|c|c|c|}
    \hline
    & Compound & AVC & FQAVC \\
    \hline
    Non-assisted & \cellcolor{gray!30} $C_\fnc{comp}$ & $C_\fnc{avc}\in\{0, C_\fnc{avc}^\fnc{SR}\}$ & $C_\fnc{fqavc}\in\{0, C_\fnc{fqavc}^\fnc{SR}\}$~\cite{boche2018fully}\\
    \hline
    Shared-randomness & \cellcolor{gray!30} $C_\fnc{comp}^\fnc{SR}$ & \cellcolor{gray!30} $C_\fnc{avc}^\fnc{SR}=C_\fnc{comp}^\fnc{SR}$~\cite{blackwell1960capacities} & \cellcolor{gray!30} $C_\fnc{fqavc}^\fnc{SR}=C_\fnc{comp}^\fnc{SR}$~\cite{boche2018fully}\\
    \hline
    Entanglement-assisted & \cellcolor{gray!70} $C_\fnc{comp}^\fnc{EA}$ & \cellcolor{gray!70} $C_\fnc{avc}^\fnc{EA}=C_\fnc{comp}^\fnc{EA}$~\cite{boche2017entanglement}& \cellcolor{gray!70} $C_\fnc{fqavc}^\fnc{EA}=C_\fnc{comp}^\fnc{EA}$~\cite{belzig2024fully}\\
    \hline
    \end{tabular}};
    
    \path ([yshift=7pt, xshift=120pt]table.north west) edge[-latex, thick] node[above, font=\scriptsize] {non-increasing} ([yshift=7pt]table.north east);
    \path ([xshift=7pt, yshift=-15pt]table.north east) edge[-latex, thick] node[above, sloped, font=\scriptsize, pos=.37]{non-decreasing} ([xshift=7pt]table.south east);
\end{tikzpicture}
\caption{(Average Error) Classical capacities of quantum channels with different shared resources and jammers, with \emph{finite-dimensional} jamming systems.
Here, cells of the same shade have been proven to take the same values \emph{under the convexity condition of the channel variability}.
}
\label{tab:summary:finite:eve}
\end{table}

In this paper, we propose an alternative approach for studying fully quantum arbitrarily varying channels (FQAVCs), which is centered around Sion’s minimax theorem and hence referred to as the {\it minimax approach}.
Specifically, we show that designing a communication protocol (between Alice and Bob) that performs well against all possible jamming states $\sigma_\sys{E}$ is no more difficult than designing one that performs well against the worst-case jamming state.
Similar minimax techniques have been successfully applied in other quantum information tasks requiring {\it universality}, such as quantum channel simulation~\cite{cao2024quantum}.
The minimax approach recovers the finite-dimensional capacity formulas and provides a route to infinite-dimensional jammer systems for entanglement assistance and shared-randomness-assisted CQ channels (\cf~\cite{boche2018fully, belzig2024fully}).

\smallskip
The remainder of the paper is organized as follows. 
In Section~\ref{sec:preliminary}, we give the definitions of compound channels, arbitrarily varying channels (AVCs), and fully quantum arbitrarily varying channels (FQAVCs), along with their classical capacities. 
In Section~\ref{sec:minimax}, we introduce the minimax approach, which is the core technical innovation of the paper. 
In Section~\ref{sec:ea:comp:fqavc}, based on the minimax approach, we provide a streamlined proof of the previously established single-letter expression of the entanglement-assisted classical capacity of FQAVCs.
In Section~\ref{sec:sr:comp:fqavc}, we demonstrate how such minimax approach is also useful when the shared resource is restricted to classical randomness, which results in a single-letter expression of the shared-randomness assisted capacities of CQ-FQAVCs with a quantum adversary system, and a multi-letter expression for QQ-FQAVCs.
In Section~\ref{sec:inf:dim}, we extend the entanglement-assisted and CQ shared-randomness-assisted results to infinite-dimensional jamming systems by introducing a trimming method to approximate the jamming system using finite-dimensional subspaces.
Finally, in Section~\ref{sec:discussion}, we summarize our findings, outline directions for further work.
\section{Preliminaries}\label{sec:preliminary}
We consider a quantum channel $\mathcal N_{\sys{AE}\to\sys B}$ with a sender input $\sys A$, a jammer input $\sys E$, and a receiver output $\sys B$.
The systems $\sys A$ and $\sys B$ are finite-dimensional. The jammer system is finite-dimensional through Section~\ref{sec:sr:comp:fqavc}, and may be infinite-dimensional in Section~\ref{sec:inf:dim}.
For infinite-dimensional systems, states are positive trace-class operators of trace one, and channels act on trace-class operators.
Write $d_{\sys A}=\dim\hilbert_{\sys A}$ and $d_{\sys B}=\dim\hilbert_{\sys B}$. All logarithms are to base two, and $h$ denotes the binary entropy.
For a channel $\mathcal T_{\sys A\to\sys B}$, we use the unnormalized Choi operator
\begin{equation}
J(\mathcal T)\defeq(\id_{\sys A'}\tensor\mathcal T)(\proj{\gamma}_{\sys A'\sys A}),
\qquad \ket{\gamma}_{\sys A'\sys A}=\sum_{i=1}^{d_{\sys A}}\ket{i}_{\sys A'}\ket{i}_{\sys A}.
\end{equation}
Thus $\tr J(\mathcal T)=d_{\sys A}$ and $\tr_{\sys B}J(\mathcal T)=I_{\sys A'}$.
We write $\norm{X}_{\tr}=\tr\sqrt{X^\herm X}$.
For a state $\rho$ and a positive operator $\sigma$, we use the quantum relative entropy $D(\rho\Vert\sigma)=\tr[\rho(\log\rho-\log\sigma)]$, with value $+\infty$ when the support condition fails, and the divergences
\begin{align}
D_h^\epsilon(\rho\Vert\sigma)&=-\log\inf\{\tr(Q\sigma):0\le Q\le I,\ \tr(Q\rho)\ge1-\epsilon\},\\
D_{\max}(\rho\Vert\sigma)&=\inf\{\lambda\in\mathbb R:\rho\le2^\lambda\sigma\}.
\end{align}

The jammer knows the code but has no access to the message, channel inputs, channel outputs, or the realization of the shared resource.
The following three models differ in the allowed states on $\sys E^n$.
\begin{itemize}
\item In the \emph{compound} model, the jammer chooses one state $\sigma_{\sys E}$ and supplies $\sigma_{\sys E}^{\tensor n}$. The induced channel is $[\mathcal N(\cdot\tensor\sigma_{\sys E})]^{\tensor n}$, with the same unknown channel used throughout the block.
\item In the \emph{arbitrarily varying} model, the jammer may supply any state separable across the $n$ inputs. Equivalently for worst-case average error, the jammer may choose an arbitrary sequence of local states. Mixtures of such sequences do not increase the supremum of the error, which is affine in the jammer state.
\item In the \emph{fully quantum arbitrarily varying} model, the jammer may supply any $\sigma_{\sys E^n}\in\densop(\hilbert_{\sys E}^{\tensor n})$, including states entangled across uses. The induced channel is
\begin{equation}
\rho_{\sys A^n}\longmapsto\mathcal N^{\tensor n}(\rho_{\sys A^n}\tensor\sigma_{\sys E^n}).
\end{equation}
\end{itemize}
A classical sender input is obtained by restricting the input algebra to operators diagonal in a fixed basis.
We study classical communication over these channels with shared randomness or shared entanglement.
A communication protocol consists of an encoder $\mathcal{E}$ (on the Alice side), a decoder $\mathcal{D}$ (on the Bob side), and a shared resource.
In the case with \emph{entanglement-assistance}\footnote{The notion of shared-randomness-assistance follows similarly by restricting systems $\sys{K}$ and $\sys{K}'$ below to classical systems.}, such a code can be formulated as a triple $(\mathcal{E},\mathcal{D},\proj{\phi}_{\sys{K}'\sys{K}})$ where
\begin{itemize}
    \item The quantum systems $\sys{K}'$ and $\sys{K}$ have the same state space, \ie, $\hilbert_{\sys{K}'} = \hilbert_{\sys{K}}$, and $\ket{\phi}_{\sys{K}'\sys{K}}$ is a pure state on systems $\sys{K}'\sys{K}$ describing the entangled system pair between Alice and Bob, who have access to system $\sys{K}'$ and $\sys{K}$, respectively.
    \item the encoder is a CPTP map from the systems $\sys{M}\sys{K}'$ to $\syss{A}_1^n$ where $\sys{M}$ is classical with basis $\{\ket{m}\}_{m=1}^{M}$, \ie, $\mathcal{E}\in\fnc{CPTP}_{[M]\tensor \hilbert_{\sys{K}'} \to \syss{A}_1^n}$, and we further denote $\mathcal{E}^{(m)}\equiv \mathcal{E}(\proj{m}\tensor \cdot)$.
    \item the decoder is a joint measurement on systems $\syss{B}_1^n$ and $\sys{K}$ described by some positive operator valued measurement (POVM) $\mathcal{D}=\left\{D^{(m)}_{\syss{B}_1^n\sys{K}}\right\}_{m=1}^M$.
    \item $M$ is the message size, and $[M]=\{1,2,\ldots,M\}$.
    \item the code uses the channel $n$ times.
\end{itemize}
For a jammer state $\sigma_{\sys E^n}$, write $\mathcal N_\sigma^{[n]}=\mathcal N^{\tensor n}(\cdot\tensor\sigma_{\sys E^n})$.
The conditional decoding probabilities are
\begin{equation}
p_{\hat{\rv M}|\rv M}[\sigma](\hat m|m)
=\tr\left[D^{(\hat m)}_{\sys B^n\sys K}
\left((\mathcal N_\sigma^{[n]}\circ\mathcal E^{(m)})\tensor\id_{\sys K}\right)(\proj\phi_{\sys K'\sys K})\right].
\end{equation}
Let $\set S_n$ be the allowed set of jammer states for the chosen model.
For uniformly distributed messages, the worst-case average error is
\begin{equation}\label{eq:worst:eps:def}
\epsilon\defeq\sup_{\sigma\in\set S_n}\left\{1-\frac1M\sum_{m=1}^M p_{\hat{\rv M}|\rv M}[\sigma](m|m)\right\}.
\end{equation}
In this case, we refer to the code $(\mathcal{E},\mathcal{D},\proj{\phi}_{\sys{K}'\sys{K}})$ as a size-$M$ code with error at most $\epsilon$ for $n$ uses of the channel $\mathcal{N}_{\sys{AE}\to\sys{B}}$, or a $(M,\epsilon)$-code for $\mathcal{N}_{\sys{AE}\to\sys{B}}^{\tensor n}$ for short.

We study the tradeoff between message size and worst-case average error as the blocklength grows.
We compare FQAVCs with the corresponding AVC and compound models.
We use quantum states to describe Eve's input in all three models, with the appropriate product or separability restriction.
For a FQAVC, AVC, or compound channel induced by the CPTP map $\mathcal{N}_{\sys{AE}\to\sys{B}}$, in all three cases, for an entanglement-assisted code $(\mathcal{E},\mathcal{D},\proj{\phi}_{\sys{K}'\sys{K}})$ over such a channel,
\begin{itemize}
    \item For FQAVC $\mathcal{N}_{\sys{AE}\to\sys{B}}$, the conditional decoding probability for a jammer state $\sigma_{\syss{E}_1^n}$ is
    \begin{align}\label{eq:fqavc:pmm}
        p_{\hat{\rv{M}}|\rv{M}}[\sigma_{\syss{E}_1^n}](\hat{m}|m) 
        &= \tr\left( D_{\syss{B}_1^n\sys{K}}^{(\hat{m})} \cdot \left[\mathcal{N}_{\sys{AE}\to\sys{B}}^{\tensor n}(\cdot\tensor \sigma_{\syss{E}_1^n})\circ \mathcal{E}_{\sys{K}'\to\syss{A}_1^n}^{(m)}\right](\proj{\phi}_{\sys{K}'\sys{K}})\right).
    \end{align}
    The worst-case average error is
    \begin{equation}\label{eq:worst:eps:fqavc}
    \epsilon \defeq \sup_{\sigma_{\syss{E}_1^n}\in\densop(\hilbert_{\syss{E}_1^n})} \left\{ 1-\sum_{m\in[M]} \frac{1}{M} \cdot p_{\hat{\rv{M}}|\rv{M}}[\sigma_{\syss{E}_1^n}](m|m) \right\}.
    \end{equation}
    \item For AVC $\mathcal{N}_{\sys{AE}\to\sys{B}}$, Eve's input is restricted to some separable states, for example $\sigma_{\syss{E}_1^n} = \int p(\mathbf{e}_1^n)\cdot \tensor_{j=1}^n \sigma_{\sys{E}_j}^{(e_j)} \diff{\mathbf{e}_1^n}$.
    The conditional decoding probabilities are again given by~\eqref{eq:fqavc:pmm}.
    The worst-case average error is
    \begin{equation}\label{eq:worst:eps:avc}
    \epsilon \defeq \sup_{\sigma_{\syss{E}_1^n}\in\densop(\hilbert_{\syss{E}_1^n}): \text{ separable}} \left\{ 1-\sum_{m\in[M]} \frac{1}{M} \cdot p_{\hat{\rv{M}}|\rv{M}}[\sigma_{\syss{E}_1^n}](m|m) \right\}.
    \end{equation}
    \item For the compound channel, Eve supplies $\sigma_{\sys E}^{\tensor n}$ for a single unknown state $\sigma_{\sys E}$. Thus
    \begin{equation}\label{eq:comp:pmm}
    p_{\hat{\rv M}|\rv M}[\sigma_{\sys E}^{\tensor n}](\hat m|m)
    =\tr\left[D^{(\hat m)}_{\sys B^n\sys K}
    \left(\bigl[\mathcal N(\cdot\tensor\sigma_{\sys E})\bigr]^{\tensor n}\circ\mathcal E^{(m)}\tensor\id_{\sys K}\right)(\proj{\phi}_{\sys K'\sys K})\right].
    \end{equation}
    The worst-case average error is
    \begin{equation}\label{eq:worst:eps:comp}
    \epsilon\defeq\sup_{\sigma_{\sys E}\in\densop(\hilbert_{\sys E})}
    \left\{1-\frac1M\sum_{m=1}^M p_{\hat{\rv M}|\rv M}[\sigma_{\sys E}^{\tensor n}](m|m)\right\}.
    \end{equation}
\end{itemize}
For each model, the entanglement-assisted classical capacity is
\begin{equation}
    C_\fnc{type}^\fnc{EA}(\mathcal{N}_{\sys{AE}\to\sys{B}}) \defeq \sup\left\{ r\geq 0 \middle\vert
    \parbox{9cm}{$\exists $ a sequence of entanglement-assisted $(\ceil{2^{nr}}, \epsilon_n)$ codes $(\mathcal{E},\mathcal{D},\proj{\phi}_{\sys{K}'\sys{K}})$ for $\mathcal{N}^{\tensor n}_{\sys{AE}\to\sys{B}}$ with $\epsilon_n\to 0$}
    \right\}.
\end{equation}
where for $\fnc{type}\in\{\fnc{fqavc},\fnc{avc},\fnc{comp}\}$, $\epsilon_n$ is given by~\eqref{eq:worst:eps:fqavc} or~\eqref{eq:worst:eps:avc} or~\eqref{eq:worst:eps:comp}, respectively.

It is immediate that 
\begin{equation}
C_\fnc{comp}^\fnc{EA}(\mathcal{N}_{\sys{AE}\to\sys{B}})
\geq C_\fnc{avc}^\fnc{EA}(\mathcal{N}_{\sys{AE}\to\sys{B}})
\geq C_\fnc{fqavc}^\fnc{EA}(\mathcal{N}_{\sys{AE}\to\sys{B}}),
\end{equation}
The inequalities follow by inclusion of the allowed jammer states.

\section{The minimax approach}
\label{sec:minimax}
In this section, we demonstrate how Sion's minimax theorem can be used for analyzing the channel coding problems with an adversary quantum side input.
In particular, we show that requiring a code to perform well for \emph{all} input states $\sigma_\sys{E}$ is no more demanding than requiring it to perform well for the \emph{worst-case} input. To facilitate our discussion, we introduce the following notations and definitions first.
\paragraph{Channels on classical-quantum systems.}
For finite classical alphabets $\set X$ and $\set Y$, we identify classical states with density operators diagonal in the corresponding computational bases.
A CPTP map from $\rv X\sys A$ to $\rv Y\sys B$ maps states of the form $\sum_x p(x)\proj{x}\tensor\rho_{\sys A}^{(x)}$ to states of the form $\sum_y q(y)\proj{y}\tensor\sigma_{\sys B}^{(y)}$.
We denote this class by $\fnc{CPTP}_{\rv X\sys A\to\rv Y\sys B}$.
Randomized codes are convex mixtures of encoder-decoder pairs. In the finite-dimensional space of joint input-output maps, finite mixtures suffice by Carath\'eodory's theorem.
\paragraph{One-shot message sizes and infimum errors.}
For $*\in\{\text{unassisted},\fnc{SR},\fnc{EA}\}$, let $\epsilon_M^*(\mathcal T)$ be the infimum average error over all $M$-message codes with the specified resource, and define
\begin{equation}
M_\epsilon^*(\mathcal T)\defeq\max\{M\in\mathbb N:\epsilon_M^*(\mathcal T)\le\epsilon\}.
\end{equation}
We omit the superscript in the unassisted case.
For a channel with a jammer, $\epsilon_M^{*,\forall\sys E}(\mathcal N)$ is the infimum, over codes, of the worst-case average error over all jammer states, and $M_\epsilon^{*,\forall\sys E}$ is defined by the same threshold convention.
This convention allows an infimum that is not attained when the entanglement dimension is unrestricted. Operationally, a threshold $\epsilon_M^*\le\epsilon$ guarantees a code with error at most $\epsilon+\xi$ for every $\xi>0$.
Such arbitrarily small slack does not change any capacity or fixed-error asymptotic statement below. For unassisted and SR codes in finite dimensions the code sets are compact, so the infima are attained.

Using the above notation, and referring to Section~\ref{sec:preliminary}, the largest code size for an $n$-fold FQAVC can be expressed as
\begin{equation}
M^{\forall\syss{E}_1^n}_{\epsilon}(\mathcal{N}^{\tensor n}_{\sys{AE}\to\sys{B}}), \text{ or }
M^{\fnc{SR}, \forall\syss{E}_1^n}_{\epsilon}(\mathcal{N}^{\tensor n}_{\sys{AE}\to\sys{B}}), \text{ or }
M^{\fnc{EA}, \forall\syss{E}_1^n}_{\epsilon}(\mathcal{N}^{\tensor n}_{\sys{AE}\to\sys{B}}),
\end{equation}
depending on the shared resources, where this simplification does not hold in general for the compound or AVC cases\footnote{In the one-shot setting, the compound, arbitrarily varying, and fully quantum arbitrarily varying channels derived from the same CPTP map $\mathcal{N}_{\sys{AE}\to\sys{B}}$ are equivalent.
Differences only arise when considering multiple uses of the channel.}.
Via the definition, we immediately have the upper bound
\begin{equation}\label{eq:M:fqavc:inf:eac}
M^{*, \forall\sys{E}}_{\epsilon}(\mathcal{N}_{\sys{AE}\to\sys{B}}) \leq \inf_{\sigma_\sys{E}} M^{*}_\epsilon(\mathcal{N}_{\sys{AE}\to\sys{B}}(\cdot\tensor\sigma_\sys{E})),
\end{equation}
where $*\in\{\text{`'}, \text{`SR'}, \text{`EA'}\}$.

We show that~\eqref{eq:M:fqavc:inf:eac} is tight for SR and EA assistance under the infimum-error convention above.
This holds because the set of shared randomness-assisted/entanglement-assisted codes of a fixed size is convex, and the average error is affine and continuous in each of the code and jammer arguments.
\begin{theorem}\label{thm:minimax}
Given a quantum channel $\mathcal{N}_{\sys{AE}\to\sys{B}}$ with adversary system $\sys{E}$ where all involved systems are finite-dimensional, it holds that
\begin{align}
\label{eq:minimax:1}
M^{\fnc{SR}, \forall\sys{E}}_{\epsilon}(\mathcal{N}_{\sys{AE}\to\sys{B}}) &= \inf_{\sigma_\sys{E}} M^\fnc{SR}_\epsilon(\mathcal{N}_{\sys{AE}\to\sys{B}}(\cdot\tensor\sigma_\sys{E})),\\
\label{eq:minimax:2}
M^{\fnc{EA}, \forall\sys{E}}_{\epsilon}(\mathcal{N}_{\sys{AE}\to\sys{B}}) &= \inf_{\sigma_\sys{E}} M^\fnc{EA}_\epsilon(\mathcal{N}_{\sys{AE}\to\sys{B}}(\cdot\tensor\sigma_\sys{E}))
\end{align}
for all $\epsilon\in(0,1)$.
\end{theorem}
\paragraph{Codes as joint-input-output maps.}
Before presenting the proof of Theorem~\ref{thm:minimax}, we represent codes as joint input-output CPTP maps.
This notation arises naturally when describing non-signaling-assisted codes and will prove useful for our subsequent discussion.
Specifically, we define the sets of all $M$-message non-assisted, shared-randomness-assisted, and entanglement-assisted codes, respectively, by
\begin{align}
\set{C}_{M} &\defeq \left\{
    \begin{aligned}
    &\zeta_{\rv{M}\sys{B}\to\sys{A}\hat{\rv{M}}}\in \\
    &\fnc{CPTP}_{[1:M]\tensor\hilbert_\sys{B}\to\hilbert_\sys{A}\tensor[1:M]}
    \end{aligned}
    \middle\vert
    \begin{aligned}
    &\exists\ \begin{aligned}[t]
    &\mathcal{E}_{\rv{M}\to\sys{A}}\in \fnc{CPTP}_{[1:M] \to\hilbert_\sys{A}}\\
    &\mathcal{D}_{\sys{B}\to\hat{\rv{M}}}\in \fnc{CPTP}_{\hilbert_\sys{B} \to[1:M]}
    \end{aligned}\\
    &\text{s.t. }
    \zeta = \mathcal{E}_{\rv{M}\to\sys{A}}\tensor \mathcal{D}_{\sys{B}\to\hat{\rv{M}}}
    \end{aligned}
    \right\}\\
 \set{C}_{M}^{\fnc{SR}} &\defeq \left\{
    \begin{aligned}
    &\zeta_{\rv{M}\sys{B}\to\sys{A}\hat{\rv{M}}}\in \\
    &\fnc{CPTP}_{[1:M]\tensor\hilbert_\sys{B}\to\hilbert_\sys{A}\tensor[1:M]}
    \end{aligned}
    \middle\vert
    \begin{aligned}
    &\exists\ \begin{aligned}[t]
    &\text{some random variable } \rv{S} \\
    &\mathcal{E}_{\rv{MS}\to\sys{A}}\in \fnc{CPTP}_{[1:M]\tensor \set{S} \to\hilbert_\sys{A}}\\
    &\mathcal{D}_{\sys{B}\rv{S}\to\hat{\rv{M}}}\in \fnc{CPTP}_{\hilbert_\sys{B}\tensor \set{S} \to[1:M]}
    \end{aligned}\\
    &\text{s.t. }
    \zeta = \int\diff{s}\ p_\rv{S}(s) \cdot \mathcal{E}_{\rv{MS}\to\sys{A}}(\cdot\tensor \proj{s}) \tensor \mathcal{D}_{\sys{B}\rv{S}\to\hat{\rv{M}}}(\cdot\tensor \proj{s})
    \end{aligned}
    \right\}\\
\set{C}_{M}^{\fnc{EA}} &\defeq \left\{
    \begin{aligned}
    &\zeta_{\rv{M}\sys{B}\to\sys{A}\hat{\rv{M}}}\in \\
    &\fnc{CPTP}_{[1:M]\tensor\hilbert_\sys{B}\to\hilbert_\sys{A}\tensor[1:M]}
    \end{aligned}
    \middle\vert
    \begin{aligned}
    &\exists\ \begin{aligned}[t]
    &\begin{aligned}[t]\text{some quantum systems } \sys{K} \text{ and }\sys{K}' \text{ with state spaces }\\
    \hilbert_\sys{K}\!=\!\hilbert_{\sys{K'}}\text{, and some pure state } \ket{\psi}_{\sys{K}'\sys{K}}\in\hilbert_{\sys{K}'}\!\tensor\!\hilbert_{\sys{K}}\end{aligned}\\
    &\mathcal{E}_{\rv{M}\sys{K}'\to\sys{A}}\in \fnc{CPTP}_{[1:M]\tensor \hilbert_{\sys{K}'} \to\hilbert_\sys{A}}\\
    &\mathcal{D}_{\sys{BK}\to\hat{\rv{M}}}\in \fnc{CPTP}_{\hilbert_\sys{B}\tensor \hilbert_\sys{K} \to[1:M]}
    \end{aligned}\\
    &\text{s.t. }
    \zeta = \left(\mathcal{E}_{\rv{M}\sys{K}'\to\sys{A}} \tensor \mathcal{D}_{\sys{BK}\to\hat{\rv{M}}} \right)(\cdot \tensor \proj{\psi}_{\sys{K}'\sys{K}})
    \end{aligned}
    \right\}
\end{align}
Both $\set C_M^{\fnc{SR}}$ and $\set C_M^{\fnc{EA}}$ are convex because Alice and Bob can use an additional shared random variable to select a code. The SR set is the convex hull of the compact set of unassisted codes and is compact in finite dimensions. The minimax argument does not require compactness of the EA code set. 
\begin{proof}[Proof of Theorem~\ref{thm:minimax}]
Consider a quantum channel $\mathcal{N}_{\sys{AE}\to\sys{B}}$ with adversary system $\sys{E}$.
Let $M$ be a positive integer.
For each shared-randomness assisted code $\zeta\in\set{C}_M^{\fnc{SR}}$, its average error when the jamming state is $\sigma_\sys{E}$ can be expressed as
\begin{equation}\label{eq:minimax:error}
\epsilon(\zeta, \sigma_\sys{E}) = 1- \sum_{m=1}^M \frac{1}{M} \cdot \bra{m} \zeta\circ\mathcal{N}(\proj{m}\tensor \sigma_\sys{E}) \ket{m}
\end{equation}
which is affine and continuous in each argument. Here $\zeta\circ\mathcal N$ denotes the wiring of the channel output to the decoder and the encoder output to the channel input.
Hence, the minimal worst-case (\wrt Eve's jamming) average error (given the message-size constraint $M$) can be written as 
\begin{equation}\label{eq:eps:M:fqavc:1}
\epsilon^{\fnc{SR},\forall\sys{E}}_{M}(\mathcal{N}_{\sys{AE}\to\sys{B}}) =  
\adjustlimits\inf_{\zeta\in\set{C}_M^\fnc{SR}} \sup_{\sigma_\sys{E}\in\densop(\hilbert_\sys{E})} \epsilon(\zeta, \sigma_\sys{E}).
\end{equation}
Since the error is separately affine and continuous, both sets $\set{C}_M^\fnc{SR}$ and $\densop(\hilbert_\sys{E})$ are convex, and that $\densop(\hilbert_\sys{E})$ is compact,
Sion's minimax theorem applies to~\eqref{eq:eps:M:fqavc:1}, and thus,
\begin{align}\label{eq:minimax:3}
\epsilon^{\fnc{SR},\forall\sys{E}}_{M}(\mathcal{N}_{\sys{AE}\to\sys{B}})
&= \adjustlimits\sup_{\sigma_\sys{E}\in\densop(\hilbert_\sys{E})}\inf_{\zeta\in\set{C}_M^\fnc{SR}} \epsilon(\zeta, \sigma_\sys{E})
= \sup_{\sigma_\sys{E}\in\densop(\hilbert_\sys{E})} \epsilon_M^\fnc{SR}(\mathcal{N}_{\sys{AE}\to\sys{B}}(\cdot\tensor\sigma_\sys{E}))
\end{align}
where $\epsilon_M^\fnc{SR}(\mathcal{N})$ denotes the minimal average error of shared-randomness-assisted $M$-alphabet classical communication over a quantum channel $\mathcal{N}$.

To prove~\eqref{eq:minimax:1}, we have 
\begin{align}
M^{\fnc{SR}, \forall\sys{E}}_{\epsilon}(\mathcal{N}_{\sys{AE}\to\sys{B}})
&= \max\left\{m\in\mathbb{N}: \epsilon^{\fnc{SR},\forall\sys{E}}_{m}(\mathcal{N}_{\sys{AE}\to\sys{B}}) \leq \epsilon \right\}\\
&= \max\left\{m\in\mathbb{N}: \epsilon_m^\fnc{SR}(\mathcal{N}_{\sys{AE}\to\sys{B}}(\cdot\tensor\sigma_\sys{E})) \leq \epsilon \quad \forall \sigma_\sys{E}\in\densop(\hilbert_\sys{E})\right\}\\
&= \inf_{\sigma_\sys{E}\in\densop(\hilbert_\sys{E})} \max\left\{m\in\mathbb{N}: \epsilon_m^\fnc{SR}(\mathcal{N}_{\sys{AE}\to\sys{B}}(\cdot\tensor\sigma_\sys{E})) \leq \epsilon\right\}\\
&= \inf_{\sigma_\sys{E}\in\densop(\hilbert_\sys{E})} M^\fnc{SR}_\epsilon(\mathcal{N}_{\sys{AE}\to\sys{B}}(\cdot\tensor\sigma_\sys{E})).
\end{align}
Finally,~\eqref{eq:minimax:2} follows from the same argument by replacing the set $\set{C}_M^\fnc{SR}$ with $\set{C}_M^\fnc{EA}$.
\end{proof}

\section{Entanglement-assisted classical capacity of FQAVCs}
\label{sec:ea:comp:fqavc}
In this section, we prove 
\begin{equation}\label{eq:ea:comp:fqavc}
C_\fnc{fqavc}^\fnc{EA}(\mathcal{N}_{\sys{AE}\to\sys{B}}) = \underbrace{\adjustlimits\sup_{\rho_{\sys{A}'}\in\densop(\hilbert_{\sys{A}'})} \inf_{\sigma_\sys{E}\in\densop(\hilbert_\sys{E})} I(\sys{A}':\sys{B})_{\left(\id_{\sys{A}'}\tensor \mathcal{N}_{\sys{AE}\to\sys{B}}\right)(\proj{\rho}_{\sys{A}'\sys{A}}\tensor \sigma_{\sys{E}}) }}_{\defas I_{\fnc{adv.}\sys{E}}(\mathcal{N}_{\sys{AE}\to\sys{B}})}
\end{equation}
\ie, $C_\fnc{comp}^\fnc{EA}(\mathcal{N}_{\sys{AE}\to\sys{B}}) = C_\fnc{fqavc}^\fnc{EA}(\mathcal{N}_{\sys{AE}\to\sys{B}})$.
We split the proof into the following converse bound and achievability bound.
\subsection{Converse bound}
\label{sec:ea:comp:fqavc:converse}
We first establish a strong converse.
\begin{proposition}\label{prop:converse:ea:fqavc}
Given a quantum channel $\mathcal{N}_{\sys{AE}\to\sys{B}}$ with an adversary system $E$, where all involved systems are finite dimensional, 
it holds that
\begin{equation}
\limsup_{n\to\infty} \frac{1}{n}\log{M^{\fnc{EA}, \forall\syss{E}_1^n}_{\epsilon}(\mathcal{N}^{\tensor n}_{\sys{AE}\to\sys{B}})} \leq I_{\fnc{adv.}\sys{E}}(\mathcal{N}_{\sys{AE}\to\sys{B}})
\end{equation}
for all $\epsilon\in(0,1)$.
\end{proposition}
\begin{proof}
Restricting the jammer to an i.i.d. state $\sigma_{\sys E}^{\tensor n}$ gives
\begin{equation}
\frac{1}{n} \log{M_\epsilon^{\fnc{EA},\forall\syss{E}_1^n}(\mathcal{N}_{\sys{AE}\to\sys{B}}^{\tensor n})}
\leq \inf_{\sigma_\sys{E}} \frac{1}{n} \log{M_\epsilon^{\fnc{EA}}\left(\left[\mathcal{N}_{\sys{AE}\to\sys{B}}(\cdot\tensor\sigma_\sys{E})\right]^{\tensor n}\right)}.
\end{equation}
For each fixed $\sigma_\sys{E}$, the strong converse theorem of the entanglement-assisted capacity~\cite{gupta2015multiplicativity} dictates that 
\begin{equation}
\lim_{n\to\infty} \frac{1}{n} \log{M_\epsilon^{\fnc{EA}}\left(\left[\mathcal{N}_{\sys{AE}\to\sys{B}}(\cdot\tensor\sigma_\sys{E})\right]^{\tensor n}\right)}
= \sup_{\rho_\sys{A}} I(\sys{A}':\sys{B})_{\id_{\sys{A}'}\tensor\mathcal{N}_{\sys{AE}\to\sys{B}}(\proj{\rho}_{\sys{A}'\sys{A}}\tensor\sigma_\sys{E})}
\end{equation}
for all $\epsilon\in(0,1)$.
Hence, 
\begin{equation}
\limsup_{n\to\infty} \frac{1}{n} \log{M_\epsilon^{\fnc{EA},\forall\syss{E}_1^n}(\mathcal{N}_{\sys{AE}\to\sys{B}}^{\tensor n})}
\leq \adjustlimits\inf_{\sigma_\sys{E}} \sup_{\rho_\sys{A}} I(\sys{A}':\sys{B})_{\id_{\sys{A}'}\tensor\mathcal{N}_{\sys{AE}\to\sys{B}}(\proj{\rho}_{\sys{A}'\sys{A}}\tensor\sigma_\sys{E})}
= I_{\fnc{adv.}\sys{E}}(\mathcal{N}_{\sys{AE}\to\sys{B}})
\end{equation}
for all $\epsilon\in(0,1)$.
\end{proof}
\subsection{Achievability bound}
\label{sec:ea:comp:fqavc:achievability}
For achievability, we combine the minimax identity with a one-shot entanglement-assisted coding bound expressed through hypothesis-testing divergence.
The generalized AEP of~\cite{fang2024generalized} controls the asymptotics for output states generated by correlated jammer inputs.
\begin{proposition}\label{prop:achievability:ea:fqavc}
Given a quantum channel $\mathcal{N}_{\sys{AE}\to\sys{B}}$ with an adversary system $E$, where all involved systems are finite dimensional, 
it holds that
\begin{equation}
\liminf_{n\to\infty} \frac{1}{n}\log{M^{\fnc{EA}, \forall\syss{E}_1^n}_{\eta}(\mathcal{N}^{\tensor n}_{\sys{AE}\to\sys{B}})} \geq I_{\fnc{adv.}\sys{E}}(\mathcal{N}_{\sys{AE}\to\sys{B}})
\end{equation}
for all $\eta\in(0,1)$
\end{proposition}
We use the following generalized AEP from~\cite{fang2024generalized}.
\begin{lemma}[Generalized AEP~{\cite[Theorem~25, arXiv version~4]{fang2024generalized}}] \label{lem:generalized:AEP}
Let $\{\set{A}_n\}_{n\in\mathbb{N}}$ and $\{\set{B}_n\}_{n\in\mathbb{N}}$ be two sequences of sets such that
\begin{enumerate}
\item $\set{A}_n\subset \densop(\hilbert^{\tensor n})$, $\set{B}_n\subset \posop(\hilbert^{\tensor n})$ for each $n\in\mathbb{N}$.
\item $\set{A}_n$ and $\set{B}_n$ are convex and compact for each $n\in\mathbb{N}$.
\item $\set{A}_n$ and $\set{B}_n$ are permutation invariant for each $n\in\mathbb{N}$.
\item $\set{A}_m\tensor\set{A}_k\subset\set{A}_{m+k}$ and $\set{B}_m\tensor\set{B}_k\subset\set{B}_{m+k}$ for each $m,k\in\mathbb{N}$.
\item $(\set{A}_m)_+^\circ \tensor(\set{A}_k)_+^\circ \subset (\set{A}_{m+k})_+^\circ$ and $(\set{B}_m)_+^\circ \tensor (\set{B}_k)_+^\circ \subset (\set{B}_{m+k})_+^\circ$ for each $m,k\in\mathbb{N}$, where for any convex subset $\set{C}\subset\linop(\hilbert)$
    \begin{equation}\label{eq:generalized:AEP:assumption:5}
    (\set{C})_+^\circ \defeq \{X\in\posop(\hilbert) \vert \tr(XY)\leq 1 \ \forall Y\in\set{C} \}.
    \end{equation}
\item $D_{\max}\infdiv*{\set{A}_n}{\set{B}_n}\defeq \inf_{\rho\in\set{A}_n,\, \sigma\in\set{B}_n} D_{\max}\infdiv*{\rho}{\sigma}$ grows at most linearly \wrt $n$.
\end{enumerate}
Then
\begin{equation}
    \lim_{n\to\infty} \frac{1}{n} \inf_{\rho\in\set{A}_n,\, \sigma\in\set{B}_n} D_h^\epsilon\infdiv*{\rho}{\sigma} = 
    \lim_{n\to\infty} \frac{1}{n} \inf_{\rho\in\set{A}_n,\, \sigma\in\set{B}_n} D\infdiv*{\rho}{\sigma}
\end{equation}
for any $\epsilon\in(0,1)$.
\end{lemma}
\begin{proof}[Proof of Proposition~\ref{prop:achievability:ea:fqavc}]
Fix an input state $\rho_{\sys A}$ with purification $\psi_{\sys R\sys A}=\proj{\rho}_{\sys R\sys A}$, and write
\begin{equation}
\omega_n(\sigma)\defeq(\id_{\sys R^n}\tensor\mathcal N^{\tensor n})(\psi_{\sys R\sys A}^{\tensor n}\tensor\sigma_{\sys E^n}).
\end{equation}
The minimax identity and the one-shot coding bound of~\cite[Theorem~1]{anshu2018building} imply, for $\epsilon,\delta>0$ with $2\epsilon+2\delta<1$,
\begin{equation}\label{eq:n-shot:EA:achievability}
\log M_{2\epsilon+2\delta}^{\fnc{EA},\forall\sys E^n}(\mathcal N^{\tensor n})
\ge \inf_{\sigma_{\sys E^n}}D_h^\epsilon\infdiv{\omega_n(\sigma)}{\rho_{\sys R}^{\tensor n}\tensor\omega_n(\sigma)_{\sys B^n}}+2\log\delta-1.
\end{equation}
The final one-bit loss accommodates integer message sizes.
Define the two channel-image sets
\begin{align}
\set A_n&=\{\omega_n(\sigma):\sigma\in\densop(\hilbert_{\sys E}^{\tensor n})\},\\
\set B_n&=\{\rho_{\sys R}^{\tensor n}\tensor\omega_n(\sigma)_{\sys B^n}:\sigma\in\densop(\hilbert_{\sys E}^{\tensor n})\}.
\end{align}
Allowing the two arguments to vary independently gives
\begin{equation}\label{eq:rate:ach:Dh:2}
\log M_{2\epsilon+2\delta}^{\fnc{EA},\forall\sys E^n}(\mathcal N^{\tensor n})
\ge \inf_{\omega\in\set A_n,\,\tau\in\set B_n}D_h^\epsilon\infdiv{\omega}{\tau}+2\log\delta-1.
\end{equation}
Both families are convex, compact, permutation invariant, and closed under tensor products.
For either channel $T:\sys E\to\sys R\sys B$ defining these image sets and every $X\ge0$,
\begin{equation}\label{eq:channel:image:polar}
\sup_{\sigma\in\densop(\hilbert_{\sys E}^{\tensor n})}\tr[X T^{\tensor n}(\sigma)]
=\norm{(T^*)^{\tensor n}(X)}_\infty.
\end{equation}
Multiplicativity of the operator norm on positive tensor products proves the positive-polar condition in Lemma~\ref{lem:generalized:AEP}.
Fixing any one-use jammer state gives an output $\omega_{\sys R\sys B}$ with
$\operatorname{supp}\omega_{\sys R\sys B}\subseteq\operatorname{supp}(\rho_{\sys R}\tensor\omega_{\sys B})$.
Thus $c=D_{\max}\infdiv{\omega_{\sys R\sys B}}{\rho_{\sys R}\tensor\omega_{\sys B}}<\infty$, and tensor powers of this pair show that $D_{\max}\infdiv{\set A_n}{\set B_n}\le nc$.

Set $q(\rho)\defeq\inf_{\sigma_{\sys E}}I(\sys R:\sys B)_{\omega_1(\sigma)}$.
Every $\omega\in\set A_n$ has marginal $\omega_{\sys R^n}=\rho_{\sys R}^{\tensor n}$, so
\begin{equation}\label{eq:lower:Dh:7}
D\infdiv{\omega}{\rho_{\sys R}^{\tensor n}\tensor\tau_{\sys B^n}}
\ge I(\sys R^n:\sys B^n)_\omega
\ge\sum_{i=1}^n I(\sys R_i:\sys B_i)_\omega\ge nq(\rho).
\end{equation}
The middle inequality follows from the chain rule, data processing, and independence of the reference systems.
Applying Lemma~\ref{lem:generalized:AEP} to~\eqref{eq:rate:ach:Dh:2}, with $\epsilon=\delta=\eta/4$, gives
\begin{equation}\label{eq:proof:achievability:ea:fqavc:last}
\liminf_{n\to\infty}\frac1n\log M_\eta^{\fnc{EA},\forall\sys E^n}(\mathcal N^{\tensor n})
\ge\sup_{\rho_{\sys A}}q(\rho)=I_{\fnc{adv.}\sys E}(\mathcal N).
\end{equation}
The equality of the sup-inf and inf-sup forms follows from concavity of channel mutual information in the input state and convexity in the channel, by Sion's theorem.
For any rate strictly below this value, apply the bound at errors $\eta_k=2^{-k}$ and choose increasing blocklength thresholds. Choosing a code at the $k$th error level between successive thresholds, with arbitrarily small additional slack if needed, yields a sequence of actual codes with vanishing worst-case error.
\end{proof}

\section{Shared-randomness-assisted classical capacities of FQAVCs}
\label{sec:sr:comp:fqavc}
In this section, we turn our attention to shared-randomness-assisted coding for classical communication over channels with a quantum jammer.
We emphasize that we do not impose any constraint on inputs from the jammer, \ie, the jammer can supply entangled jamming signals across multiple channel uses.
In the following, we consider two scenarios: the underlying channels are classical-quantum, or fully quantum (\ie, a FQAVC).
Some of the methods used in this section are similar to Section~\ref{sec:ea:comp:fqavc}, in which case we will only describe the differences.
\subsection{SR-assisted capacity of CQ-FQAVCs}
\label{sec:sr:comp:fqavc:cq}
This is a special case of FQAVC where we restrict system $\sys{A}$ to be classical.
Such a channel can be denoted by $\mathcal{N}_{\rv{X}\sys{E}\to\sys{B}}$ where we have replaced the input system $\sys{A}$ by a classical system $\rv{X}$.
With unlimited shared randomness between the sender and receiver, we prove
\begin{equation}\label{eq:cq:sr:fqavc:capacity}
C_\fnc{fqavc}^\fnc{SR}(\mathcal{N}_{\rv{X}\sys{E}\to\sys{B}}) = \adjustlimits\inf_{\sigma_{\sys{E}}\in\densop(\hilbert_\sys{E})} \sup_{p_\rv{X}\in\set{P}(\set{X})} I(\rv{X}':\sys{B})_{\mathcal{N}_{\rv{X}\sys{E}\to\sys{B}}(p_\rv{X}\cdot\delta_{\rv{X},\rv{X}'}\tensor\sigma_\sys{E})}
\end{equation}
Similar to Section~\ref{sec:ea:comp:fqavc}, we split the proof into the following converse and achievability bounds.
Both bounds are based on applying~\eqref{eq:minimax:1} on $n$-copies of the channel $\mathcal{N}_{\rv{X}\sys{E}\to\sys{B}}$, \ie, 
\begin{equation}
M^{\fnc{SR}, \forall\syss{E}_1^n}_{\epsilon}\left(\mathcal{N}^{\tensor n}_{\rv{X}\sys{E}\to\sys{B}}\right)
= \inf_{\sigma_{\syss{E}_1^n}} M^\fnc{SR}_\epsilon\left(\mathcal{N}^{\tensor n}_{\rv{X}\sys{E}\to\sys{B}}(\cdot\tensor\sigma_{\syss{E}_1^n})\right)
= \inf_{\sigma_{\syss{E}_1^n}} M_\epsilon\left(\mathcal{N}^{\tensor n}_{\rv{X}\sys{E}\to\sys{B}}(\cdot\tensor\sigma_{\syss{E}_1^n})\right),
\end{equation}
where the second equality holds due to the standard derandomization argument.
\paragraph{Strong converse.}
As in Section~\ref{sec:ea:comp:fqavc:converse}, restrict the jammer to an i.i.d. state to obtain
\begin{align}
\frac{1}{n} \log{M_\epsilon^{\fnc{SR},\forall \syss{E}_1^n}(\mathcal{N}_{\rv{X}\sys{E}\to\sys{B}}^{\tensor n})}
&\leq \frac{1}{n} \log{M_\epsilon\left(\left[\mathcal{N}_{\rv{X}\sys{E}\to\sys{B}}(\cdot\tensor\sigma_\sys{E})\right]^{\tensor n}\right)}
\end{align}
for every jammer state $\sigma_\sys{E}\in\densop(\hilbert_\sys{E})$.
Taking $n\to\infty$ and applying the strong converse for CQ channels~\cite{winter1999coding} gives 
\begin{align}
\limsup_{n\to\infty} \frac{1}{n} \log{M_\epsilon^{\fnc{SR},\forall \syss{E}_1^n}(\mathcal{N}_{\rv{X}\sys{E}\to\sys{B}}^{\tensor n})}
&\leq \inf_{\sigma_\sys{E}\in\densop(\hilbert_\sys{E})} \limsup_{n\to\infty} 
 \frac{1}{n} \log{M_\epsilon\left(\left[\mathcal{N}_{\rv{X}\sys{E}\to\sys{B}}(\cdot\tensor\sigma_\sys{E})\right]^{\tensor n}\right)} \\
&\leq \inf_{\sigma_\sys{E}\in\densop(\hilbert_\sys{E})} C\left(\mathcal{N}_{\rv{X}\sys{E}\to\sys{B}}(\cdot\tensor\sigma_\sys{E})\right)
= \text{RHS of~\eqref{eq:cq:sr:fqavc:capacity}}
\end{align}
for all $\epsilon\in(0,1)$.
\paragraph{Achievability.}
Achievability follows from the same flow of arguments as in Section~\ref{sec:ea:comp:fqavc:achievability}, and we only describe the key differences below.
In particular, by the one-shot achievability bound on CQ channel capacity as in~\cite[Eq.~(12)]{cheng2023simple} (also see~\cite{hayashi2003general} for an earlier bound that this one improved upon), we have (\cf~\eqref{eq:n-shot:EA:achievability}), 
\begin{equation}
\log{M^{\fnc{SR}, \forall\syss{E}_1^n}_{\epsilon}\left(\mathcal{N}^{\tensor n}_{\rv{X}\sys{E}\to\sys{B}}\right)} \geq 
\inf_{\sigma_{\syss{E}_1^n}} D_h^{\epsilon-\delta}\infdiv*{\mathcal{N}_{\rv{X}\sys{E}\to\sys{B}}^{\tensor n}(p_{\rv{X}'\rv{X}}^{\tensor n}\tensor \sigma_{\syss{E}_1^n})}{p_{\rv{X}'}^{\tensor n} \tensor \mathcal{N}_{\rv{X}\sys{E}\to\sys{B}}^{\tensor n}(p_\rv{X}^{\tensor n}\tensor \sigma_{\syss{E}_1^n})} - \log{\frac{1}{\delta}}-1
\end{equation}
for $0<\delta<\epsilon<1$ and any pmf $p_\rv{X}\in\set{P}(\set{X})$ where $\rv{X}'$ is a copy of $\rv{X}$, \ie, $p_{\rv{X}'\rv{X}}(x',x) = p_\rv{X}(x)\cdot \mathbbm{1}\{x=x'\}$.
For each pmf $p_\rv{X}$, let 
\begin{align}
    \set{A}_n &\defeq \left\{\mathcal{N}_{\rv{X}\sys{E}\to\sys{B}}^{\tensor n}(p_{\rv{X}'\rv{X}}^{\tensor n}\tensor \sigma_{\syss{E}_1^n})\middle\vert \sigma_{\syss{E}_1^n}\in\densop(\hilbert_{\syss{E}_1^n})\right\}, \\
    \set{B}_n &\defeq \left\{p_{\rv{X}'}^{\tensor n} \tensor \mathcal{N}_{\rv{X}\sys{E}\to\sys{B}}^{\tensor n}(p_\rv{X}^{\tensor n}\tensor \sigma_{\syss{E}_1^n})\middle\vert \sigma_{\syss{E}_1^n}\in\densop(\hilbert_{\syss{E}_1^n})\right\}.
\end{align}
Following the arguments from ~\eqref{eq:n-shot:EA:achievability} to~\eqref{eq:proof:achievability:ea:fqavc:last}, we have 
\begin{equation}
\liminf_{n\to\infty} \frac{1}{n} \log{M^{\fnc{SR}, \forall\syss{E}_1^n}_{\epsilon}\left(\mathcal{N}^{\tensor n}_{\rv{X}\sys{E}\to\sys{B}}\right)}
\geq \adjustlimits\inf_{\sigma_{\sys{E}}\in\densop(\hilbert_\sys{E})} \sup_{p_\rv{X}\in\set{P}(\set{X})} I(\rv{X}':\sys{B})_{\mathcal{N}_{\rv{X}\sys{E}\to\sys{B}}(p_\rv{X}\cdot\delta_{\rv{X},\rv{X}'}\tensor\sigma_\sys{E})}
\end{equation}
for all $\epsilon\in(0,1)$.
The interchange of the input and jammer optimizations follows from concavity of $I(\rv X':\sys B)$ in $p_{\rv X}$ and convexity in $\sigma_{\sys E}$. Vanishing error follows by choosing successively smaller fixed error levels as in Section~\ref{sec:ea:comp:fqavc:achievability}.
\subsection{SR-assisted capacity of general quantum channels}
\label{sec:sr:comp:fqavc:qq}
For general quantum inputs, define
\begin{equation}
q_n\defeq\inf_{\sigma_{\sys E^n}\in\densop(\hilbert_{\sys E}^{\tensor n})}
\sup_{\rho_{\rv X\sys A^n}}
I(\rv X:\sys B^n)_{(\id_{\rv X}\tensor\mathcal N^{\tensor n})(\rho_{\rv X\sys A^n}\tensor\sigma_{\sys E^n})},
\end{equation}
where the supremum ranges over finite ensembles of input states, represented as cq states $\rho_{\rv X\sys A^n}$.
We prove
\begin{equation}\label{eq:qq:sr:fqavc:capacity}
C_{\fnc{fqavc}}^{\fnc{SR}}(\mathcal N)=\lim_{n\to\infty}\frac{q_n}{n}=\sup_{n\ge1}\frac{q_n}{n}.
\end{equation}

We first justify the minimax interchange for ensembles.
Represent an ensemble by a probability measure $\mu$ on the compact state space $\densop(\hilbert_{\sys A}^{\tensor n})$.
For a fixed induced channel $\mathcal T$, its Holevo information is
\begin{equation}
\chi(\mu,\mathcal T)=H\left(\mathcal T\left(\int\rho\,\mathrm d\mu(\rho)\right)\right)
-\int H(\mathcal T(\rho))\,\mathrm d\mu(\rho),
\end{equation}
where $H$ is the von Neumann entropy.
This functional is concave in $\mu$, convex in the channel, and continuous in the weak topology on measures and the usual finite-dimensional topology on channels.
Sion's theorem therefore interchanges the ensemble supremum and the jammer infimum.
Finite-support measures approximate these values uniformly over the compact set of induced channels, so the same equality holds with the supremum over finite ensembles.
For independent ensembles on $m$ and $k$ uses, the product classical marginal gives
$I(\rv X_1\rv X_2:\sys B^m\sys B^k)\ge I(\rv X_1:\sys B^m)+I(\rv X_2:\sys B^k)$.
Taking the jammer infimum and then the ensemble suprema yields $q_{m+k}\ge q_m+q_k$.
Since $0\le q_n\le n\log d_{\sys B}$, Fekete's lemma proves existence of the limit in~\eqref{eq:qq:sr:fqavc:capacity}.

\paragraph{Weak converse.}
For a fixed jammer state, shared randomness can be removed by selecting a deterministic code whose average error is no greater than the randomized average.
Fano's inequality and the Holevo bound then imply
\begin{equation}
(1-\epsilon)\log M_\epsilon^{\fnc{SR},\forall\sys E^n}(\mathcal N^{\tensor n})\le q_n+h(\epsilon).
\end{equation}
Dividing by $n$ and taking a sequence $\epsilon_n\to0$ proves the converse in~\eqref{eq:qq:sr:fqavc:capacity}.

\paragraph{Achievability.}
Fix a positive integer $m$ and a finite ensemble $\{p(x),\rho_{\sys A^m}^{(x)}\}$.
Preparing $\rho_{\sys A^m}^{(x)}$ before using $\mathcal N^{\tensor m}$ defines a CQ-FQAVC with jammer input $\sys E^m$ and output $\sys B^m$.
Section~\ref{sec:sr:comp:fqavc:cq} gives an achievable rate, per block of $m$ uses, of
\begin{equation}
\inf_{\sigma_{\sys E^m}}I(\rv X:\sys B^m)_{(\id_{\rv X}\tensor\mathcal N^{\tensor m})(\rho_{\rv X\sys A^m}\tensor\sigma_{\sys E^m})}.
\end{equation}
The jammer may be entangled across these blocks, as allowed by the CQ-FQAVC model.
Optimizing over ensembles and using the minimax interchange above gives the rate $q_m/m$ per use of $\mathcal N$.
Codes for multiples of $m$ extend to all sufficiently large blocklengths by ignoring fewer than $m$ additional uses, whose effect on the rate vanishes.
Taking the supremum over $m$ completes the proof.

\section{Generalization to infinite-dimensional jamming systems}
\label{sec:inf:dim}
We now allow an infinite-dimensional jammer. The sender and receiver systems remain finite-dimensional.
We prove the entanglement-assisted formula~\eqref{eq:ea:comp:fqavc} in this setting and explain its CQ shared-randomness-assisted counterpart~\eqref{eq:cq:sr:fqavc:capacity}.
The jammer state space is no longer compact in trace norm, so the direct application of Sion's theorem on that space requires modification.
At each fixed blocklength, however, the induced channels lie in a bounded finite-dimensional space of Choi operators.
We approximate that set by restricting the jammer to a finite-dimensional subspace, control the resulting error uniformly over codes, and apply the generalized AEP to compact closures of output-state sets.
The approximating subspace may depend on both the blocklength and the desired error.

We first record an elementary finite approximation lemma.
\begin{lemma}\label{lemma:trimming}
    Let $\set{V}$ be a normed vector space over the field of real numbers.
    Let $\set{D} \defeq \Cov\left\{\mathbf{v}_1,\ldots,\mathbf{v}_K\right\}$ where $K\geq 2$ is an integer.
    Suppose that $\norm{\mathbf{v} - \mathbf{u}} \leq 1$ for any $\mathbf{v}$ and $\mathbf{u}\in\set{D}$.
    Let $\set{C}\defeq \Cov\left\{\mathbf{u}_i\middle\vert i\in\set{I}\right\}$ where $\set{I}$ is an infinite index set and $\mathbf{u}_i\in\set{D}$ for each $i\in\set{I}$.
    Then, for any $\epsilon\in(0,1)$, there exists a finite subset $\set{I}_\epsilon\subset\set{I}$ such that\footnote{Note that we use $\min$ instead of $\inf$ in~\eqref{eq:finite:approx} since the convex hull of finite points is always compact.}
    \begin{equation}\label{eq:finite:approx}
        \norm{\mathbf{v} - \set{C}_\epsilon} \defeq \min_{\mathbf{u}\in\set{C}_\epsilon} \norm{\mathbf{v}-\mathbf{u}} \leq \epsilon
        \qquad\forall \mathbf{v}\in\set{C}
    \end{equation}
    where $\set{C}_\epsilon\defeq\Cov\left\{\mathbf{u}_i\middle\vert i\in\set{I}_\epsilon\right\}$, and $\size{\set{I}_\epsilon} \leq \ceil{\frac{K-1}{\epsilon}}^{K-1}$.
\end{lemma}
\begin{proof}
Choose, for each $i\in\set I$, coefficients $\lambda_{i,j}\ge0$ with $\sum_{j=1}^K\lambda_{i,j}=1$ and $\mathbf u_i=\sum_j\lambda_{i,j}\mathbf v_j$.
The first $K-1$ coefficients lie in the simplex
\begin{equation}
\set P_{K-1}=\left\{x\in[0,1]^{K-1}:\sum_{j=1}^{K-1}x_j\le1\right\}.
\end{equation}
The affine parametrization $\iota(x)=\mathbf v_K+\sum_{j<K}x_j(\mathbf v_j-\mathbf v_K)$ satisfies
\begin{equation}\label{eq:parametrized:metric}
\norm{\iota(x)-\iota(y)}\le\sum_{j<K}|x_j-y_j|=\norm{x-y}_1.
\end{equation}
Partition $[0,1]^{K-1}$ into at most $\ceil{(K-1)/\epsilon}^{K-1}$ cells of side length at most $\epsilon/(K-1)$, assigning boundary points consistently.
For each cell containing a chosen coefficient vector, retain one corresponding index.
Every $\mathbf u_i$ is then within $\epsilon$ of a retained point $\widetilde{\mathbf u}_i$.
For any finite convex combination $\mathbf v=\sum_i p_i\mathbf u_i$, the point $\widetilde{\mathbf v}=\sum_i p_i\widetilde{\mathbf u}_i$ lies in $\set C_\epsilon$ and satisfies
\begin{equation}
\norm{\mathbf v-\widetilde{\mathbf v}}\le\sum_i p_i\norm{\mathbf u_i-\widetilde{\mathbf u}_i}\le\epsilon.
\end{equation}
The same estimate extends to the closure of $\set C$ because distance to the compact set $\set C_\epsilon$ is continuous.
\end{proof}
The induced Choi operators form a bounded subset of the finite-dimensional real vector space of Hermitian operators on $\sys A'\sys B$.
An enclosing polytope therefore exists. Applying Lemma~\ref{lemma:trimming} after rescaling the norm to make its diameter at most one yields the following consequence.
\begin{corollary}\label{cor:finite:approx:CPTP}
    Let $\mathcal{N}_{\sys{AE}\to\sys{B}}$ be a quantum channel with adversary system $\sys{E}$ where the input and output systems ($\sys{A}$ and $\sys{B}$) are finite-dimensional, and the system $\sys{E}$ is infinite-dimensional.
    For any $\epsilon\in(0,1)$, there exists a finite-dimensional subspace $\hilbert_\sys{E}^\epsilon \leq \hilbert_\sys{E}$ such that 
    \begin{equation}
    \adjustlimits \sup_{\sigma_{\sys{E}}\in\densop(\hilbert_\sys{E})} \min_{\tilde{\sigma}_\sys{E}\in\densop(\hilbert_\sys{E}^\epsilon)}
    \norm*{J(\mathcal{N}_{\sys{AE}\to\sys{B}}(\cdot\tensor\sigma_{\sys{E}})) - J(\mathcal{N}_{\sys{AE}\to\sys{B}}(\cdot\tensor\tilde{\sigma}_{\sys{E}}))}_{\tr}
    \leq \epsilon
    \end{equation}
    where $J(\mathcal{N})$ stands for the unnormalized Choi operator of CPTP $\mathcal{N}$.
\end{corollary}
\begin{proof}
Choose a finite $\epsilon$-net from the set of Choi operators induced by unit jammer vectors.
Such a net exists by total boundedness in finite dimensions.
Let $\hilbert_{\sys E}^\epsilon$ be the span of the finitely many corresponding jammer vectors.
A spectral decomposition of an arbitrary jammer state expresses its Choi operator as a countable convex combination of pure-state-induced Choi operators.
Replacing each term by its selected net point changes the trace norm by at most $\epsilon$.
Grouping terms with the same selected point gives a state supported on $\hilbert_{\sys E}^\epsilon$.
\end{proof}

\noindent In particular, the corollary above allows us to restrict the jamming system of $\mathcal{N}_{\sys{AE}\to\sys{B}}$ to a finite-dimensional subspace as an approximation.
For each $\epsilon\in(0,1)$, we denote this approximation by $\mathcal{N}_{\sys{AE}\to\sys{B}}^\epsilon$, which is a fully finite-dimensional CPTP.
\subsection{Converse Bound}\label{sec:inf:dim:converse}
Since the approximation $\mathcal{N}_{\sys{AE}\to\sys{B}}^\epsilon$ only restricts the choices of the jamming input, it is operationally immediate that
\begin{equation}
C_\fnc{fqavc}^\fnc{EA}(\mathcal{N}_{\sys{AE}\to\sys{B}})
\leq C_\fnc{fqavc}^\fnc{EA}(\mathcal{N}_{\sys{AE}\to\sys{B}}^\epsilon)
= I_{\fnc{adv.}\sys{E}}(\mathcal{N}_{\sys{AE}\to\sys{B}}^\epsilon)
= \adjustlimits\sup_{\rho_{\sys{A}'}} \inf_{\tilde{\sigma}_\sys{E}\in\densop(\hilbert_\sys{E}^\epsilon)} I(\sys{A}'\!:\!\sys{B})_{\left(\id_{\sys{A}'}\tensor \mathcal{N}_{\sys{AE}\to\sys{B}}\right)(\proj{\rho}_{\sys{A}'\sys{A}}\tensor \tilde{\sigma}_{\sys{E}})}\!.
\end{equation}
By Corollary~\ref{cor:finite:approx:CPTP}, we know that for any $\sigma_\sys{E}\in\densop(\hilbert_\sys{E})$, there exists $\tilde{\sigma}_\sys{E}\in\densop(\hilbert_\sys{E}^\epsilon)$ such that 
\begin{equation}
\norm*{J(\mathcal{N}_{\sys{AE}\to\sys{B}}(\cdot\tensor\sigma_{\sys{E}})) - J(\mathcal{N}_{\sys{AE}\to\sys{B}}(\cdot\tensor\tilde{\sigma}_{\sys{E}}))}_{\tr}
    \leq \epsilon,
\end{equation}
and hence 
\begin{equation}\begin{aligned}
&\hspace{13pt}\norm*{\left(\id_{\sys{A}'}\tensor \mathcal{N}_{\sys{AE}\to\sys{B}}\right)(\proj{\rho}_{\sys{A}'\sys{A}}\tensor \sigma_{\sys{E}}) - \left(\id_{\sys{A}'}\tensor \mathcal{N}_{\sys{AE}\to\sys{B}}\right)(\proj{\rho}_{\sys{A}'\sys{A}}\tensor \tilde{\sigma}_{\sys{E}})}_{\tr}\\
&= \norm*{\sqrt{\rho_{\sys{A}'}}J(\mathcal{N}_{\sys{AE}\to\sys{B}}(\cdot\tensor\sigma_{\sys{E}}))\sqrt{\rho_{\sys{A}'}} - \sqrt{\rho_{\sys{A}'}}J(\mathcal{N}_{\sys{AE}\to\sys{B}}(\cdot\tensor\tilde{\sigma}_{\sys{E}}))\sqrt{\rho_{\sys{A}'}}}_{\tr}
    \leq \epsilon.
\end{aligned}\end{equation}
By the continuity of the quantum mutual information, we know 
\begin{equation}
\abs*{I(\sys{A}'\!:\!\sys{B})_{\left(\id_{\sys{A}'}\tensor \mathcal{N}_{\sys{AE}\to\sys{B}}\right)(\proj{\rho}_{\sys{A}'\sys{A}}\tensor \sigma_{\sys{E}})} - I(\sys{A}'\!:\!\sys{B})_{\left(\id_{\sys{A}'}\tensor \mathcal{N}_{\sys{AE}\to\sys{B}}\right)(\proj{\rho}_{\sys{A}'\sys{A}}\tensor \tilde{\sigma}_{\sys{E}})}}
\leq 4\epsilon\cdot\log(d_\sys{A}\cdot d_\sys{B}) + 2h(\epsilon)
\end{equation}
where $d_\sys{A}\defeq \dim{\hilbert_\sys{A}}$ and $d_\sys{B}\defeq \dim{\hilbert_\sys{B}}$.
Hence, 
\begin{align}
C_\fnc{fqavc}^\fnc{EA}(\mathcal{N}_{\sys{AE}\to\sys{B}})
&\leq \adjustlimits\sup_{\rho_{\sys{A}'}} \inf_{\sigma_\sys{E}\in\densop(\hilbert_\sys{E})} I(\sys{A}':\sys{B})_{\left(\id_{\sys{A}'}\tensor \mathcal{N}_{\sys{AE}\to\sys{B}}\right)(\proj{\rho}_{\sys{A}'\sys{A}}\tensor \sigma_{\sys{E}})}
+ 4\epsilon\cdot\log(d_\sys{A}\cdot d_\sys{B}) + 2h(\epsilon) \\
& = I_{\fnc{adv.}\sys{E}}(\mathcal{N}_{\sys{AE}\to\sys{B}}) + 4\epsilon\cdot\log(d_\sys{A}\cdot d_\sys{B}) + 2h(\epsilon).
\end{align}
Since the above holds for all $\epsilon>0$, it must hold that 
\begin{equation}
C_\fnc{fqavc}^\fnc{EA}(\mathcal{N}_{\sys{AE}\to\sys{B}}) \leq I_{\fnc{adv.}\sys{E}}(\mathcal{N}_{\sys{AE}\to\sys{B}}).
\end{equation}
\subsection{Achievability Bound}
For the achievability of the channel coding for FQAVCs with infinite-dimensional jammers, we shall need a more sophisticated trimming argument for using the channel $n$ times (\cf~Corollary~\ref{cor:finite:approx:CPTP}).
\begin{lemma}\label{lem:finite:approx:CPTP:n}
    Let $\mathcal{N}_{\sys{AE}\to\sys{B}}$ be a quantum channel with adversary system $\sys{E}$ where the input and output systems ($\sys{A}$ and $\sys{B}$) are finite-dimensional, and the system $\sys{E}$ is infinite-dimensional.
    For any blocklength $n\in\integers_{>0}$ and $\varepsilon\in(0,1)$, there exists a finite-dimensional subspace $\set{F}_\sys{E} \leq \hilbert_\sys{E}$ such that 
    \begin{equation}
    \adjustlimits \sup_{\sigma_{\syss{E}_1^n}\in\densop(\hilbert_\sys{E}^{\tensor n})} \min_{\tilde{\sigma}_{\syss{E}_1^n}\in\densop(\set{F}_\sys{E}^{\tensor n})}
    \norm*{J(\mathcal{N}^{\tensor n}_{\sys{AE}\to\sys{B}}(\cdot\tensor\sigma_{\syss{E}_1^n})) - J(\mathcal{N}^{\tensor n}_{\sys{AE}\to\sys{B}}(\cdot\tensor\tilde{\sigma}_{\syss{E}_1^n}))}_{\tr}
    \leq \varepsilon.
    \end{equation}
\end{lemma}
\begin{proof}
The Choi operators induced by unit vectors in $\hilbert_{\sys E}^{\tensor n}$ form a bounded subset of a finite-dimensional space.
Choose a finite $\varepsilon/2$-net from this set and denote the corresponding unit jammer vectors by $\phi_1,\ldots,\phi_L$.
By spectral decomposition and convexity, for every jammer state $\sigma$ there is a convex combination $\widetilde\sigma=\sum_{j=1}^L p_j\proj{\phi_j}$ such that
\begin{equation}\label{eq:finite:approx:CPTP:n:1}
\norm{J(\mathcal N^{\tensor n}(\cdot\tensor\sigma))-J(\mathcal N^{\tensor n}(\cdot\tensor\widetilde\sigma))}_{\tr}\le\varepsilon/2.
\end{equation}
Finite linear combinations of elementary tensors are dense in $\hilbert_{\sys E}^{\tensor n}$.
Consequently, the finitely many vectors $\phi_j$ admit simultaneous approximation by normalized projections $\widetilde\phi_j$ onto $\set F_{\sys E}^{\tensor n}$ for a common finite-dimensional $\set F_{\sys E}\subset\hilbert_{\sys E}$.
For unit vectors $u,v$, trace-norm contractivity of the channel and the Choi normalization give
\begin{equation}
\norm{J(\mathcal N^{\tensor n}(\cdot\tensor\proj u))-J(\mathcal N^{\tensor n}(\cdot\tensor\proj v))}_{\tr}
\le d_{\sys A}^{\,n}\norm{\proj u-\proj v}_{\tr}
\le2d_{\sys A}^{\,n}\norm{u-v}.
\end{equation}
Choose the projections so that $\norm{\phi_j-\widetilde\phi_j}\le\varepsilon/(4d_{\sys A}^{\,n})$ for every $j$.
Then replacing $\widetilde\sigma$ by $\widehat\sigma=\sum_jp_j\proj{\widetilde\phi_j}\in\densop(\set F_{\sys E}^{\tensor n})$ costs at most another $\varepsilon/2$.
The triangle inequality proves the assertion, and compactness of $\densop(\set F_{\sys E}^{\tensor n})$ ensures attainment of the minimum.
\end{proof}

\begin{remark}\label{rem:error:continuity}
Lemma~\ref{lem:finite:approx:CPTP:n} controls the change in error uniformly over codes at fixed blocklength.
For a joint input-output map $\zeta_{\rv M\sys B^n\to\sys A^n\hat{\rv M}}$, define
\begin{align}
\Pi_{\mathrm{corr}}&=\sum_{m=1}^M\proj m_{\rv M}\tensor\proj m_{\hat{\rv M}},\\
\tau_{\sys A^n\sys B^n}&=\frac1M\left[\tr_{\rv M\hat{\rv M}}\bigl(\Pi_{\mathrm{corr}}J(\zeta)\bigr)\right]^{\trans}.
\end{align}
Identity operators on unlisted systems are implicit, and $\trans$ is the full transpose on $\sys A^n\sys B^n$ in the Choi bases.
The success probability is the Choi pairing
\begin{equation}
1-\epsilon(\zeta,\sigma)=\tr\left[\tau J(\mathcal N^{\tensor n}(\cdot\tensor\sigma))\right].
\end{equation}
Each selected message block corresponds to a completely positive, trace-nonincreasing map from $\sys B^n$ to $\sys A^n$.
Hence $\tau\ge0$ and $\tr_{\sys A^n}\tau\le I_{\sys B^n}$, which imply $\norm\tau_\infty\le\min\{d_{\sys A},d_{\sys B}\}^n$.
Indeed, for a top eigenvector $v$ with eigenvalue $\lambda$, one has $\lambda v_{\sys B^n}\le I_{\sys B^n}$, while the rank of $v_{\sys B^n}$ is at most $\min\{d_{\sys A},d_{\sys B}\}^n$.
H\"older's inequality and Lemma~\ref{lem:finite:approx:CPTP:n} therefore give
\begin{equation}\label{eq:bound:error:finite}
0\le\sup_{\sigma\in\densop(\hilbert_{\sys E}^{\tensor n})}\epsilon(\zeta,\sigma)
-\sup_{\sigma\in\densop(\set F_{\sys E}^{\tensor n})}\epsilon(\zeta,\sigma)
\le\min\{d_{\sys A},d_{\sys B}\}^n\varepsilon.
\end{equation}
The upper bound is independent of both the code and its message size.
\end{remark}

Fix $\eta\in(0,1)$ and a purified input $\psi_{\sys R\sys A}$.
For each $n$, choose $\set F_{\sys E}^{(n)}$ using Lemma~\ref{lem:finite:approx:CPTP:n} with $\varepsilon=(\eta/2)\min\{d_{\sys A},d_{\sys B}\}^{-n}$.
Let $\mathcal N^{(n)}$ denote the one-use channel with $\sys E$ restricted to this subspace, and define $\omega_n(\sigma)$ as in Section~\ref{sec:ea:comp:fqavc:achievability}.
Remark~\ref{rem:error:continuity} and~\eqref{eq:n-shot:EA:achievability}, with $\epsilon=\delta=\eta/8$, yield
\begin{align}\label{eq:M:n:finite:dim:restrict}
\log M_\eta^{\fnc{EA},\forall\sys E^n}(\mathcal N^{\tensor n})
&\ge\log M_{\eta/2}^{\fnc{EA},\forall\sys E^n}((\mathcal N^{(n)})^{\tensor n})\\
&\ge\inf_{\sigma\in\densop((\set F_{\sys E}^{(n)})^{\tensor n})}
 D_h^{\eta/8}\infdiv{\omega_n(\sigma)}{\rho_{\sys R}^{\tensor n}\tensor\omega_n(\sigma)_{\sys B^n}}+2\log(\eta/8)-1\\
&\ge\inf_{\sigma\in\densop(\hilbert_{\sys E}^{\tensor n})}
 D_h^{\eta/8}\infdiv{\omega_n(\sigma)}{\rho_{\sys R}^{\tensor n}\tensor\omega_n(\sigma)_{\sys B^n}}+2\log(\eta/8)-1.
\end{align}
Define the trace-norm closures
\begin{align}\label{eq:inf:dim:An}
\set A_n&=\operatorname{cl}\{\omega_n(\sigma):\sigma\in\densop(\hilbert_{\sys E}^{\tensor n})\},\\
\set B_n&=\operatorname{cl}\{\rho_{\sys R}^{\tensor n}\tensor\omega_n(\sigma)_{\sys B^n}:\sigma\in\densop(\hilbert_{\sys E}^{\tensor n})\}.
\end{align}
Allowing the two arguments to vary independently over these larger sets gives
\begin{equation}\label{eq:inf:dim:one:shot}
\log M_\eta^{\fnc{EA},\forall\sys E^n}(\mathcal N^{\tensor n})
\ge\inf_{\omega\in\set A_n,\,\tau\in\set B_n}D_h^{\eta/8}\infdiv{\omega}{\tau}+2\log(\eta/8)-1.
\end{equation}
The closures are compact because the output space is finite-dimensional. They remain convex, permutation invariant, and closed under tensor products.
To check the positive-polar condition, define
\begin{equation}
\Phi(\sigma)=(\id_{\sys R}\tensor\mathcal N)(\psi_{\sys R\sys A}\tensor\sigma),
\qquad \Psi(\sigma)=\rho_{\sys R}\tensor\mathcal N(\rho_{\sys A}\tensor\sigma).
\end{equation}
For $T=\Phi$ or $T=\Psi$, let $\set C_n=\operatorname{cl}T^{\tensor n}(\densop(\hilbert_{\sys E}^{\tensor n}))$.
Trace-norm continuity of the pairing with a bounded operator gives
\begin{equation}
X\in(\set C_n)_+^\circ\quad\Longleftrightarrow\quad
X\ge0\ \text{and}\ (T^*)^{\tensor n}(X)\le I_{\sys E^n}.
\end{equation}
Taking tensor products of these positive operator inequalities proves the required inclusion even when $\hilbert_{\sys E}$ is infinite-dimensional.
The $D_{\max}$ growth condition follows from tensor powers of any fixed one-use jammer output and its product of marginals, exactly as in Section~\ref{sec:ea:comp:fqavc:achievability}.

For every $\omega\in\set A_n$, its reference marginal remains $\rho_{\sys R}^{\tensor n}$, and each one-site marginal belongs to $\set A_1$.
Finite-dimensional continuity of mutual information therefore extends~\eqref{eq:lower:Dh:7} to these closed sets, giving
\begin{equation}
\inf_{\omega\in\set A_n,\,\tau\in\set B_n}D\infdiv{\omega}{\tau}
\ge n\inf_{\sigma_{\sys E}} I(\sys R:\sys B)_{\Phi(\sigma_{\sys E})}.
\end{equation}
Apply Lemma~\ref{lem:generalized:AEP} to~\eqref{eq:inf:dim:one:shot}, divide by $n$, and optimize over the fixed input state to obtain
\begin{equation}
\liminf_{n\to\infty}\frac1n\log M_\eta^{\fnc{EA},\forall\sys E^n}(\mathcal N^{\tensor n})
\ge I_{\fnc{adv.}\sys E}(\mathcal N),\qquad 0<\eta<1.
\end{equation}
For every $r<I_{\fnc{adv.}\sys E}(\mathcal N)$, choose error levels $\eta_k=2^{-k}$ and increasing thresholds $N_k$ beyond which rate $r$ is achievable at infimum error $\eta_k$.
On $N_k\le n<N_{k+1}$, choose an actual code with error at most $2\eta_k$.
This proves achievability with vanishing error against the full jammer and completes the proof of~\eqref{eq:ea:comp:fqavc}.

\paragraph{CQ channels with shared randomness.}
The same extension applies to~\eqref{eq:cq:sr:fqavc:capacity}.
Fix an input distribution and replace $\psi_{\sys R\sys A}$ by the perfectly correlated classical state $p_{\rv X'\rv X}$.
The code-uniform approximation in Remark~\ref{rem:error:continuity} applies to SR codes, and the CQ one-shot bound gives the analogue of~\eqref{eq:inf:dim:one:shot}.
The two closed output families satisfy the same AEP assumptions, with the fixed reference marginal $p_{\rv X'}^{\tensor n}$.
Optimizing over $p_{\rv X}$ gives the single-letter sup-inf formula.
Sion's theorem interchanges the optimizations using compactness of the finite input probability simplex, so compactness of the jammer state space is unnecessary for this interchange.
The converse follows by restricting the jammer to any fixed i.i.d. state and using the CQ strong converse.
Finally, the EA and CQ formulas equal their compound-channel counterparts. Compound codes face a subset of the FQAVC jammer states, while fixed-channel converses and the same minimax interchanges bound the compound capacities by these formulas.

\section{Conclusion}\label{sec:discussion}
We formulated assisted coding against a quantum jammer as a minimax problem over joint input-output maps and jammer states.
The resulting one-shot identity, combined with hypothesis-testing coding bounds and the generalized AEP, yields the entanglement-assisted capacity, the shared-randomness-assisted CQ capacity, and a regularized shared-randomness-assisted capacity formula for general quantum inputs.
For entanglement assistance and CQ channels with shared randomness, finite-dimensional approximation of the induced channels extends the results to infinite-dimensional jammer systems.
The proof permits the approximating subspace to grow with blocklength and uses a uniform operational error bound before taking asymptotic limits.

Several questions remain. Extending capacity dichotomy and derandomization results to infinite-dimensional jammer systems requires further analysis.
For general quantum inputs, it is also natural to ask when the regularized SR formula admits a simpler expression.
Quantitative convergence bounds for the generalized AEP may lead to finite-blocklength refinements, provided their constants can be controlled uniformly over the relevant jammer outputs.
Quantum communication under the same adversarial model is another direction for the minimax method.

\section*{Acknowledgments}
We acknowledge funding from the European Research Council (ERC Grant Agreement No.~948139) and the Excellence Cluster Matter and Light for Quantum Computing (ML4Q).
We would like to thank Paula Belzig for discussions on her work~\cite{belzig2024fully}.
We would also like to thank Milán Mosonyi for pointing out the non-compactness of $\densop(\hilbert_\sys{E})$ in the infinite-dimensional case, and Li Gao for further discussions on the topic.

\section*{Declaration of the use of ChatGPT}
ChatGPT (OpenAI), including the Codex interface, was used during the preparation and revision of this manuscript to discuss mathematical arguments, identify possible proof gaps and notation inconsistencies, suggest and refine proof steps, rewrite the abstract, and edit the prose and LaTeX source.
The principal mathematical assistance concerned the second half of Section 6.2, specifically applying the authors’ existing finite-dimensional approximation lemmas at the finite-blocklength stage of the argument underlying Proposition 3, before taking the asymptotic limit. This included assistance with the uniform coding-error bound in Remark 8 and the application of the generalized asymptotic equipartition property to closed output-state sets.
The authors retain responsibility for the mathematical claims, references, and final text.

\bibliographystyle{IEEEtran}
\bibliography{reference}

\end{document}